\documentclass[12pt,a4paper]{article}
\usepackage[english]{babel}
\usepackage{graphicx}
\usepackage{booktabs}
\usepackage{hyperref}
\hypersetup{colorlinks,urlcolor=blue,linkcolor=blue,citecolor=blue}
\usepackage{dcolumn}
\newcolumntype{d}[1]{D{.}{.}{#1}}
\usepackage[round]{natbib}
\usepackage{bbm}
\usepackage{multirow}
\usepackage{authblk}
\usepackage{color}
\usepackage{rotating}
\usepackage{amsmath, amsthm, amssymb}
\usepackage[utf8]{inputenc}
\usepackage[left=2.5cm,right=2.5cm,top=2.5cm,bottom=2.5cm]{geometry}
\usepackage[doublespacing]{setspace}
\usepackage[flushleft]{threeparttable}
\usepackage{placeins}
\usepackage{appendix}
\usepackage{caption}
\usepackage{float} 
\usepackage{amssymb}
\usepackage{bm}
\usepackage{enumitem}

\newtheorem{theorem}{Theorem}[section]
\newtheorem{proposition}[theorem]{Proposition}
\newtheorem{corollary}{Corollary}[theorem]

\newtheorem{algorithm}[theorem]{Algorithm}

\newcommand*\rfrac[2]{{}^{#1}\!/_{#2}}

\title{\singlespacing \bf A Spatial Fay-Herriot Model when the Auxiliary Variables are based on Non-traditional Data}

\author[1]{Robin\ Markwitz}
\author[1]{Angelo\ Moretti}
\author[1]{Camilla\ Salvatore}
\affil[1]{Department of Methodology and Statistics, Utrecht University, P.O. Box 80140, 3508 TC Utrecht, The Netherlands}
\date{}

\begin{document}

\maketitle

\begin{abstract}
\singlespacing
\noindent
Small area estimation methods combine direct survey estimates with model-based predictions to produce reliable estimates of population quantities. When covariates are measured with error, as often occurs when auxiliary information is coming from big data sources, standard area-level approaches such as the traditional Fay–Herriot model can perform poorly if this measurement error is ignored, potentially yielding biased estimates. Building on the extension proposed by Ybarra and Lohr to account for measurement error in covariates, we further develop this modelling framework by incorporating spatial dependence between areas. We propose a spatial measurement error Fay–Herriot model that jointly accounts for covariate measurement error and spatial correlation, enabling information borrowing across relevant neighbouring areas while properly adjusting for the additional uncertainty introduced by the covariates. We derive the properties of the proposed model and its parameter estimators and outline a parametric bootstrap procedure for estimating the mean squared error of the resulting empirical best linear unbiased predictor. A simulation study examines model performance under a range of measurement error scenarios. The proposed approach is then applied to European Social Survey data, supplemented with big data auxiliary information, to estimate regional attitudes towards climate change in Spain.
\end{abstract}


\newpage
\section{Introduction}

In recent years, small area estimation (SAE) has attracted increasing attention, largely in response to growing practical and policy-oriented needs for reliable sub-national statistics. In this context, the term ``small area'' generally refers to a specific subpopulation or domain for which direct estimates derived solely from observations within that domain in the survey sample are unreliable or unavailable because the corresponding sample size is too limited or even equal to zero. Small areas may be defined in various ways depending on the objectives of the study. Examples include geographically defined units such as regions, states, provinces, counties, or municipalities. Alternatively, they may refer to demographic or socio-economic groups characterised by combinations of socio-economic variables, such as age, sex, ethnicity, or other characteristics \citep{Rao15}.

The demand for reliable information at such disaggregated levels has increased substantially in recent decades, particularly in contexts where policy interventions and resource allocation decisions depend on detailed knowledge of local conditions. For example, public administrations frequently require accurate and precise estimates for small areas to guide the distribution of funding, to design targeted social policies, or to monitor inequalities across regions and communities. Consequently, the development of statistical methodologies capable of producing reliable estimates for small domains has become an important area of research within official statistics and applied statistical science \citep{pratesi2016poverty}.

The rapid development of digital technologies has led to the generation of large volumes of information about human activities. In addition to data produced through traditional administrative systems, individuals leave digital traces through their interactions with online platforms, mobile devices, and social media, that can be used for research \citep{keusch2021digital}. Similarly, citizen-science projects, in which volunteers actively contribute to data collection, have become an important source of data \citep{ponti2020citizen,pratesi2021citizen}. These new forms of data have considerably expanded the range of information available to researchers and official statistical institutes for studying socio-economic phenomena.

Such large and complex datasets, commonly referred to as Big Data, often capture information from a substantial proportion of the population within a given geographical area. Big Data sources offer important advantages, including high geographical and temporal granularity, and the possibility of measuring phenomena that would otherwise be difficult or prohibitively expensive to observe through conventional data collection systems. For these reasons, Big Data has attracted growing interest within the SAE literature.

Two main approaches have emerged for incorporating these non-traditional data sources into SAE models \citep{marchetti2015bigdata}.
In the first approach, the outcome variable of interest is derived directly from the Big Data source. In this case, it is important to recognise that Big Data are, usually, not collected through a probabilistic sampling design and data production is outside of the researchers' control. Consequently, such data are highly selective and may systematically exclude part of the target population. Researchers must therefore address the selection bias that arises from the non-probability nature of the data source. Within the SAE framework, \citet{iacus2020controlling} propose a strategy to reduce selection bias in a subjective well-being indicator derived from Twitter (now X) data through the inclusion of weights in a spatio-temporal SAE model. More recently, \citet{schirripa2025inference} proposed an approach based on the integration of probability and non-probability samples to reduce selection bias when predicting local-level indicators for enterprises, assuming that the variable of interest is observed only in the Big Data source.

The second approach uses information extracted from Big Data as auxiliary variables within SAE models. In this setting, Big Data-derived variables serve as \textit{proxies} for relevant individual- or area-level characteristics that explain regional variability and improve the precision of small area estimates. Several studies have demonstrated the potential of this strategy. \citet{porter2014spatial} incorporate Google Trends data as covariates in a spatial Fay-Herriot (FH) model to estimate relative changes in the percentage of household Spanish-speakers in the eastern half of the United States. \citet{marchetti2015bigdata} employed mobility indicators derived from GPS-tracked vehicle movements as auxiliary information in a modified FH model to estimate poverty indicators for local areas in Tuscany. Similarly, \cite{marchetti2016use} use information extracted from Twitter to construct an indicator of well-being, namely iHappy, and include it as a covariate in an FH model to predict household food consumption expenditure across Italian provinces. \cite{schmid2017constructing} exploit mobile phone data to derive sub-national estimates of the share of illiterates disaggregated by gender in Senegal. 

The advantage of using Big Data covariates lies in the fact that such variables are potentially able to explain both sub-national variability of the phenomenon of interest, and, as opposed to traditional sources, are available at a granular level.
While they can provide valuable information, their availability is subject to certain constraints, for example data access limitations, volatility over time (e.g., the output depends on when the data are accessed), and errors when transforming unstructured data (e.g., social media messages) into structured data (e.g., sentiment scores). Thus, to some extent, they are affected by a form of measurement error.

In this article, we focus on the second approach. Existing methodologies have  addressed measurement error and spatial dependence separately. For methodological developments in spatial modelling approaches in SAE, we refer the reader to the comprehensive reviews by \cite{pratesi2008small} and \cite{Petrucci06}. In many practical applications involving Big Data sources, both features are present simultaneously: auxiliary variables are often affected by measurement error, while the target phenomenon exhibits spatial dependence. If ignored, this can lead to invalid inference.

This article is motivated by this practical research challenge and proposes a spatial measurement error Fay-Herriot model that jointly accounts for measurement error in auxiliary variables and spatial dependence between areas. The proposed model extends the framework of \cite{Ybarra08} by incorporating spatial dependence, thereby enabling information borrowing across neighbouring areas while properly adjusting for measurement error in Big Data-derived covariates. The remainder of the paper is organised as follows. Section \ref{sec:FH} introduces the Fay-Herriot model and its spatial extension. Section \ref{Sec:spatialFH} reviews the Fay-Herriot model with measurement error and extends it to account for spatial dependence. Section \ref{sec:simulation} presents a simulation study to evaluate the properties of the proposed model. Section \ref{sec:application} illustrates the proposed method through an application that combines European Social Survey data and Big Data derived from Booking.com to study attitudes towards climate change in Spain. Section \ref{sec:conclusions} concludes and discusses directions for future research.

\section{The Fay-Herriot model}\label{sec:FH}

Assume that we have area-level data of the form $U = U_1, \dots, U_D$, where each $U_d$ for $d = 1, \dots, D$ is disjoint. Our aim is to estimate a target parameter $Y$ over the population $U$. Here, we consider small area means $\bm{y} = (y_1, \dots, y_D)$ over the corresponding areas $U_1, \dots, U_D$. We assume that these small area means follow a Fay-Herriot model \citep{FayHerriot79}, since we consider area-level data. 

In the Fay-Herriot modelling framework, one typically has 
\begin{equation}
\hat{\bm{y}} = \bm{y} + \bm{s},
\end{equation}
with the idea being that $\hat{\bm{y}}$ is a design-unbiased direct estimator for $\bm{y}$ with known sampling error $\bm{s}$. We assume that $\bm{s} \sim \mathcal{N}(0, \bm{R})$ where $\bm{R} = \sigma_{s}^{2}\bm{I}$. The quantity $\bm{y}$ is in turn defined as

\begin{equation}
    \bm{y} = \bm{X}\bm{\beta} + \bm{Z}\bm{v},
\end{equation}

where $\bm{X}$ is a $D \times P$ matrix of covariate data such that $\bm{X} = (\bm{X}_1, \dots, \bm{x}_d)^T$ where $\bm{x}_d = (x_{d1}, \dots, x_{dP})$ for $d = 1, \dots, D$, $\bm{\beta} = (\beta_1, \dots, \beta_P)^T$ is a $P \times 1$ parameter vector, $\bm{Z}$ is a $D \times D$ matrix of known constants and $\bm{v}$ is a $D \times 1$ vector denoting the random effects, taking the between-area variation into account.

The random effects are assumed to be independently and identically distributed as
\[
\bm{v} \sim \mathcal{N}(\bm{0}, \sigma_v^2\bm{I}),
\]
and independent of the sampling errors $\bm{s}$. Combining the sampling and linking models yields the linear mixed model

\begin{equation}
\hat{\bm{y}} = \bm{X}\bm{\beta} + \bm{Z}\bm{v} + \bm{s},
\end{equation}

with
\[
\mathbb{E}(\hat{\bm{y}})=\bm{X}\bm{\beta}
\]
and
\[
\mathrm{var}(\hat{\bm{y}})
=\bm{V}
=\sigma_v^2\bm{Z}\bm{Z}^{T}+\bm{R}.
\]

Under the assumption that the variance component is known, the best linear unbiased predictor (BLUP) of the small area means is given by
\begin{equation}
\tilde{\bm{y}}
=\bm{X}\tilde{\bm{\beta}}
+\sigma_v^2\bm{Z}\bm{Z}^{T}\bm{V}^{-1}
\left(\hat{\bm{y}}-\bm{X}\tilde{\bm{\beta}}\right),
\end{equation}
where
\[
\tilde{\bm{\beta}}
=
(\bm{X}^{T}\bm{V}^{-1}\bm{X})^{-1}
\bm{X}^{T}\bm{V}^{-1}\hat{\bm{y}}
\]
is the generalized least squares (GLS) estimator of $\bm{\beta}$. In practice, the variance component $\sigma_v^2$ is unknown and must be estimated from the data, typically using maximum likelihood (ML) or restricted maximum likelihood (REML). Replacing $\sigma_v^2$ by an estimator $\hat{\sigma}_v^2$ yields the empirical best linear unbiased predictor (EBLUP),
\begin{equation}\label{eq:eblupFH}
\hat{\bm{y}}^{EBLUP}
=
\bm{X}\hat{\bm{\beta}}
+
\hat{\sigma}_v^2\bm{Z}\bm{Z}^{T}\hat{\bm{V}}^{-1}
\left(
\hat{\bm{y}}-\bm{X}\hat{\bm{\beta}}
\right),
\end{equation}
where
\[
\hat{\bm{V}}
=
\hat{\sigma}_v^2\bm{Z}\bm{Z}^{T}
+
\bm{R}.
\]

The EBLUP can be interpreted as a weighted average of the direct estimator and the regression-synthetic estimator, where the weights depend on the relative magnitude of the sampling variance and the between-area variance. As a result, areas with large sampling variances borrow more strength from the auxiliary information contained in $\bm{X}$, whereas areas with more precise direct estimates rely less on the model-based component.

The mean squared error (MSE) of \ref{eq:eblupFH} can be estimated analytically \citep{datta2000unified, Prasad90}.

\subsection{Spatial extension}
\label{sec:spat_extension}
We wish to model the random effect $\bm{v}$ via a spatial process. One way to capture spatial effects involves letting $\bm{v}$ regress simultaneously on itself, taking into account some normally distributed error. This is known as a simultaneous autoregressive (SAR) model, and is given by
\begin{equation}
    \bm{v} = \bm{B}\bm{v} + \bm{\varepsilon},
\end{equation}
where $\bm{B}$ is a $D \times D$ matrix accounting for spatial dependence in the second moment and $\bm{\varepsilon} \sim \mathcal{N}(0,\sigma_v^2\bm{I})$ \citep{Lieshout19, VerHoef18}. Usually, $\bm{B} = \rho \bm{W}$, where $\rho$ is an autocorrelation parameter and $\bm{W}$ is the proximity matrix. Throughout this work, we allow $\bm{W}$ to be non-symmetric. We therefore have $\bm{v} \sim \mathcal{N}(\bm{0}, \bm{G})$ where
\begin{equation}
    \bm{G} =\text{cov}(\bm{v}) = \sigma_v^2((\bm{I} - \rho\bm{W})(\bm{I} - \rho\bm{W}^T))^{-1}.
\end{equation}
We need to guarantee that $(\bm{I} - \rho \bm{W})^{-1}$ exists (note that this is sufficient for the whole product inverse to exist). For SAR models, this is the case when $\rho \notin \{\lambda_d^{-1}; d = 1, \dots, D \}$ \citep{Li12, VerHoef18}, where $\lambda_{i}$ is the $i$th eigenvalue of $\bm{W}$. This means that we need to ensure that $-1 < \rho < 1$. Note that if this condition is satisfied, $\bm{G}$ is positive definite. 

Similarly to the case of the traditional Fay--Herriot model, the EBLUP (spatial EBLUP) can be written as follows (notice that we now have the $\bm{G}$ matrix as well): 

\begin{equation}\label{eq:seblup}
\hat{\bm{y}}^{SEBLUP}
=
\bm{X}\hat{\bm{\beta}}
+
\bm{Z}\hat{\bm{G}}\bm{Z}^{T}\hat{\bm{V}}^{-1}
\left(
\hat{\bm{y}}
-
\bm{X}\hat{\bm{\beta}}
\right),
\end{equation}

where
\[
\hat{\bm{V}}
=
\bm{Z}\hat{\bm{G}}\bm{Z}^{T}
+
\bm{R},
\]
and
\[
\hat{\bm{G}}
=
\hat{\sigma}_v^2
\left(
(\bm{I}-\hat{\rho}\bm{W})
(\bm{I}-\hat{\rho}\bm{W}^{T})
\right)^{-1}.
\]

The spatial EBLUP (SEBLUP) borrows strength not only through the auxiliary information contained in $\bm{X}$ but also through the spatial dependence structure given in $\bm{W}$. Consequently, neighbouring areas contribute information to one another, potentially improving the precision of estimates.

The MSE of \ref{eq:seblup} can be estimated analytically \citep{singh2005spatio} or via bootstrap \citep{molina2009bootstrap}. These types of spatial models have been extensively studied in small area estimation, and for more details on this we refer to \cite{pratesi2008small} and \cite{Petrucci06}.

\section{The spatial Fay-Herriot model with measurement error in the covariates}\label{Sec:spatialFH}

When using non-traditional data sources as auxiliary variables, it may be the case that some of them are measured with error. In that case, we only have knowledge of an estimator $\hat{\bm{X}}$ of $\bm{X}$. \citet{Ybarra08} proposed the model
\begin{equation}\label{eq:YLModel}
    \hat{y}_d = \hat{\bm{x}}_d\bm{\beta} + r_d(\hat{\bm{x}_d}, \bm{x}_d) + s_d,
\end{equation}
where 
\begin{equation}
    r_d(\hat{\bm{x}_d}, \bm{x}_d) = v_d + (\bm{x}_d - \hat{\bm{x}}_d)\bm{\beta}.
\end{equation}
Independence assumptions, i.e., $\hat{\bm{X}}$ and $\hat{\bm{y}}$ are independent, are in order. In \citet{Ybarra08}, the assumption $v_d \sim \mathcal{N}(0, \sigma_v^2)$ is made. In this section, we use the notation as in \cite{Ybarra08}. 

\subsection{Model definition}
\label{sec:model_def}
We investigate the case when the random effects $\bm{v}$ are spatially correlated, i.e. $\bm{v} \sim \mathcal{N}(\bm{0}, \bm{G})$, and the covariates are potentially measured with error. We will refer to this model construction as the spatial Fay-Herriot model with measurement error. In this case, the model takes the same form as in equation \ref{eq:YLModel}, however $r_d$ becomes

\begin{equation}
    r_d(\hat{\bm{x}_d}, \bm{x}_d) = \bm{z}_d\bm{v} + (\bm{x}_d - \hat{\bm{x}_d})\bm{\beta}
\end{equation}
where $\bm{z}_d$ is the $d$th row of $\bm{Z}$, taking into account the spatial extension. The full model is

\begin{equation}
        \hat{y}_d = \hat{\bm{x}}_d\bm{\beta} + (\bm{x}_d - \hat{\bm{x}_d})\bm{\beta} + \bm{z}_d\bm{v} + s_d.
        \label{eq:spat_yl_formulation}
\end{equation}
Note that $\text{MSE}(\hat{\bm{x}_d}) = \bm{C}_d =  \mathbb{E}[(\bm{x}_d - \hat{\bm{x}_d})^T(\bm{x}_d - \hat{\bm{x}_d})]$, as in \citet{Ybarra08}. 

We start by estimating the MSE of $r_d$ in the spatial case. 

\begin{proposition}[MSE of $r_d$]
\label{thm:mse_r}
    Assume that $\hat{\bm{X}}$ is an unbiased estimator of $\bm{X}$. Then \[\text{MSE}\,(r_d) = \bm{z}_d\bm{G}\bm{z}_d^T + \bm{\beta}^T\bm{C}_d\bm{\beta}.\]
\end{proposition}
\begin{proof}
    See appendix.
\end{proof}
\begin{corollary}
    Defining $\bm{C}_{\beta} = \text{diag}_{1\leq d \leq D}(\bm{\beta}^T\bm{C}_d\bm{\beta})$, we can rewrite the MSE in matrix notation as 
    \begin{align*}
        \text{MSE}(\bm{r}) = \bm{Z}\bm{G}\bm{Z}^T + \bm{C}_{\beta}
    \end{align*}
    where $\bm{r}$ is a $D \times 1$ vector.
\end{corollary}
We now continue by finding the empirical best linear unbiased predictor (EBLUP) of this spatial measurement error model. The EBLUP is a convex linear combination of the direct survey estimate $\hat{\bm{y}}$ and the mean $\mathbb{E}[\hat{\bm{y}}] = \hat{\bm{X}}^T\hat{\bm{\beta}}$ \citep{Henderson50, Rao15} and is given in the following theorem.

\begin{theorem}[Spatial measurement error EBLUP]
\label{thm:seblup}
    Assume that $\mathbb{E}[\hat{\bm{X}}] = \bm{X}$, \newline $\text{MSE}\,(\hat{\bm{x}}_d) = \bm{C}_d$, $\text{MSE}\,(r_d) = \bm{z}_d\bm{G}\bm{z}_d^T + \bm{\beta}^T\bm{C}_d\bm{\beta}$ and $(\hat{\bm{x}}_i, \bm{z}_i, s_i)$ is independent of $(\hat{\bm{x}}_j, \bm{z}_j, s_j)$ for any $i,j = 1, \dots, D.$ Additionally assume that $\hat{\bm{x}}_d$, $\bm{z}_d$ and $s_d$ are independent of each other for $d = 1, \dots, D$. Replace $\bm{\beta}, \sigma_v^2, \rho$ with estimators $\hat{\bm{\beta}}, \hat{\sigma}_v^2, \hat{\rho}$, assuming these exist. The EBLUP becomes
    
\begin{equation}\label{eq:spatial_EBLUPYL}
\begin{split}
\hat{y}^{SEBLUP,ME}_{d}
=&\,
\hat{\bm{x}}_d\hat{\bm{\beta}}
+\bm{z}_d \bm{b}_d^T
\left(
\hat{\sigma}_v^2
\bigl(
(\bm{I}-\hat{\rho}\bm{W})
(\bm{I}-\hat{\rho}\bm{W}^T)
\bigr)^{-1}
\bm{Z}^T
+\bm{C}_{\bm{\beta}}
\right) \\
&\times
\Bigl[
\bm{R}
+
\bm{Z}\hat{\sigma}_v^2
\bigl(
(\bm{I}-\hat{\rho}\bm{W})
(\bm{I}-\hat{\rho}\bm{W}^T)
\bigr)^{-1}
\bm{Z}^T
+\bm{C}_{\bm{\beta}}
\Bigr]^{-1}
(\hat{\bm{y}}-\hat{\bm{X}}\hat{\bm{\beta}}).
\end{split}
\end{equation}

where $\bm{b}_d^T = (0, \dots, 0, 1, 0, \dots, 0)$ is of dimension $1 \times D$ with $1$ in the $d$th position, keeping consistent with the notation in \citet{Petrucci06}.  
\end{theorem}
\begin{proof}
    See appendix.
\end{proof}
\noindent
When there is no measurement error, $\bm{C}_{\bm{\beta}}$ is the zero matrix and the EBLUP reduces to the standard spatial EBLUP as in \citet{Petrucci06}. Note that when $\bm{C}_{\bm{\beta}} - \bm{Z}\hat{\bm{G}}\bm{Z}^T$ is positive definite, using the classical spatial EBLUP becomes worse than using the direct estimator, as \citet{Ybarra08} showed for the standard Fay-Herriot model. Not adding this correction term leads to underestimation of the MSE when errors in $\bm{X}$ are ignored.

\subsection{Log-likelihood and derivatives}
\label{sec:loglik_deriv}
We provide the full model formulation in Section~\ref{sec:model_def}, equation~\ref{eq:spat_yl_formulation}. Note that in most cases, $\bm{x}_d - \hat{\bm{x}_d}$, the difference between the actual, unknown covariate values and the estimates, is not known a priori. In practice, it is usually assumed that $\bm{x}_d - \hat{\bm{x}_d} = \bm{u}_d$ where $\bm{u}_d \sim \mathcal{N}(0, \bm{C}_d)$ and the $\bm{C}_d$ are known \citep{Burgard20}. Recall that $\bm{C}_d = \text{MSE}(\hat{\bm{x}}_d)$ where $d$ in $1, \dots, D$, $\bm{C}_{\bm{\beta}} = \text{diag}_{1\leq d \leq D}(\bm{\beta}^T\bm{C}_d\bm{\beta})$ and also let $\bm{U} = [\bm{u}_d]_{d=1}^D$. Assuming normality allows us to easily derive estimators and their properties. Note that the $\bm{C_d}$ may be entirely unknown - estimation approaches are beyond the scope of this article. We can rewrite the model from equation~\ref{eq:spat_yl_formulation} as
\begin{equation}
    \hat{\bm{y}} = (\hat{\bm{X}} + \bm{U})\bm{\beta} + \bm{Z}\bm{v} + \bm{s}.
    \label{eq:yl_model_practice}
\end{equation} 

We know, e.g. via proposition~\ref{thm:mse_r}, that $\text{var}(\bm{U}\bm{\beta}) = \bm{C}_{\bm{\beta}}$. Conditionally on $\hat{\bm{X}}$, we hence have the following model structure:
\begin{align*}
    \hat{\bm{y}}\,|\,\hat{\bm{X}} &\sim \mathcal{N}(\hat{\bm{X}}\bm{\beta}, \bm{V}), \\
    \bm{V}(\bm{\beta}, \sigma_v^2, \rho) &= \bm{R} + \bm{Z}\bm{G}\bm{Z}^T + \bm{C}_{\bm{\beta}}, \\
    \bm{G}(\sigma_v^2, \rho) &= \sigma_v^2((\bm{I} - \rho\bm{W})(\bm{I} - \rho\bm{W}^T))^{-1}.
\end{align*}
Write $\bm{e}$ for the residual $\hat{\bm{y}} - \hat{\bm{X}}\hat{\bm{\beta}}$ and recall that $\mathbb{E}(\bm{e}) = 0$. Let the parameter vector be $\bm{\theta} = (\bm{\beta}, \sigma_v^2, \rho)$ and let $\hat{\bm{\beta}}, \hat{\sigma}_v^2, \hat{\rho}$ be their estimators. We also note that as long as $\bm{V}, \bm{M}, \bm{G}$ are invertible, any function $f$ that takes values over any non-empty subset of $\bm{\theta}$ is continuous. Since $\hat{\bm{y}}$ is multivariate normally distributed, the log-likelihood is
\begin{equation}
    \ell(\bm{\theta}) = -\frac{1}{2}\left(D\log(2\pi) + \log|\bm{V}| + \bm{e}^T\bm{V}^{-1}\bm{e}\right),
    \label{eq:loglik}
\end{equation}
see, e.g. \citet{Duchateau04}. 

We continue by taking derivatives. Note that due to the extra term depending on $\bm{\beta}$ in $\bm{V}$, differentiation is not identical to the standard case. We hence do not use the usual generalised least squares (GLS) estimator for $\bm{\beta}$, but instead obtain a variance-corrected estimator via maximum likelihood (ML). The loss of one degree of freedom due to using the GLS estimator has led to restricted maximum likelihood (REML) procedures being favoured in the literature. As we choose not to do this in this case, the utility of applying a REML approach is lost. In fact, this will simply lead to larger MSE, see discussions in \citet{Harville77} and \citet{Vasdekis05} for more details. Therefore, we provide ML estimators for the aforementioned parameters. 
\begin{proposition}[Derivative w.r.t. $\beta$]
\label{thm:derivative_beta}
    Let $\bm{k} = \bm{V}^{-1}\bm{e}$, hence $\mathbb{E}(\bm{k}) = 0$ and $\mathbb{E}(\bm{k}\bm{k}^T) = \bm{V}^{-1}$. The derivative of $\ell(\bm{\theta})$ with respect to $\bm{\beta}$ is
    \begin{equation}
         \frac{\partial\ell}{\partial\bm{\beta}} = \hat{\bm{X}}^T\bm{k} + \sum_{d=1}^D [\bm{k}_d^2 - (\bm{V}^{-1})_{dd}]\bm{C}_d\bm{\beta}.
        \label{eq:derivative_beta}
    \end{equation}
\end{proposition}
\begin{proof}
    See appendix.
\end{proof}
In the standard case without measurement error, the $\bm{\beta}$-derivative only contains the first term, and can be solved directly for $\bm{\beta}$ (GLS estimator). The second term in this model can be interpreted as a variance correction term in the presence of measurement error.

The scalar derivatives with respect to $\sigma_v^2$ and $\rho$ do not change from the standard spatial Fay-Herriot model, and are provided in \citet{pratesi2008small}. For completeness, we give them below. Note that we do not require symmetric $\bm{W}$, meaning that the derivatives may be somewhat different to other forms found in other works. 

Let $\bm{M} = (\bm{I} - \rho\bm{W})(\bm{I} - \rho\bm{W}^T)$. Now $\bm{G} = \sigma_v^2\bm{M}^{-1}$ and hence $\bm{V} = \bm{R} + \sigma_v^2\bm{Z}\bm{M}^{-1}\bm{Z}^{T} + \bm{C}_{\beta}$. See that $\bm{M}'_{\rho} = 2\rho \bm{W}\bm{W}^T - (\bm{W}+\bm{W}^T)$ and $(\bm{M}^{-1})_{\rho}^{'} = \bm{M}^{-1}\bm{M}'_{\rho}\bm{M}^{-1}$.

Using this, we obtain the derivatives $\bm{V}'_{\sigma} = \bm{Z}\bm{M}^{-1}\bm{Z}^T$ and $\bm{V}'_{\rho}=-\sigma_v^2\bm{Z}(\bm{M}^{-1})_{\rho}^{'}\bm{Z}^T$. Hence
\begin{align*}
    \frac{\partial l}{\partial \hat{\sigma}_v^2} = -\frac{1}{2}\text{tr}(\bm{V}^{-1}\bm{V}'_{\sigma}) + \frac{1}{2}\bm{e}^T(\bm{V}^{-1}\bm{V}'_{\sigma}\bm{V}^{-1})\bm{e}
\end{align*}
and
\begin{align*}
    \frac{\partial l}{\partial \hat{\rho}} &= -\frac{1}{2}\text{tr}(\bm{V}^{-1}\bm{V}'_{\rho}) + \frac{1}{2}\bm{e}^T(\bm{V}^{-1}\bm{V}'_{\rho}\bm{V}^{-1})\bm{e}.
\end{align*}
Note that this system (including the derivative w.r.t. $\bm{\beta}$) cannot be solved directly and the ML estimators are not available in closed form, since all parameters are contained in $\bm{V}$ and the derivative equation~\ref{eq:derivative_beta} is not linear in $\bm{\beta}$. We therefore proceed by outlining a parameter estimation mechanism.

\subsection{Parameter estimation}
\label{sec:par_est}

The aim of parameter estimation is to obtain $\bm{\theta}$, the parameter vector. This can be found by solving the equation $\nabla_{\theta}\ell(\bm{\theta}) = 0$, after which we obtain an estimate $\hat{\bm{\theta}}$. To do this, we use Fisher scoring \citep{Longford87}. For multivariate normal distributions with parameter vector $\bm{\theta} = (\theta_1, \dots, \theta_P, \theta_{P+1}, \theta_{P+2})$ where $\theta_{P+1} = \sigma_v^2$ and $\theta_{P+2} = \rho$, observed quantity vector $\bm{y} = (y_1, \dots, y_D)$ and mean $\hat{\bm{X}}\bm{\beta}$, the Fisher information is given by
\begin{align}
    \mathcal{I}(\bm{\theta})_{m,n} = \frac{\partial(\hat{\bm{X}}\bm{\beta}
    )^T}{\partial\bm{\theta}_m} \bm{V}^{-1} \frac{\partial(\hat{\bm{X}}\bm{\beta}
    )}{\partial\bm{\theta}_n} + \frac{1}{2}\text{tr}\left(\bm{V}^{-1}(\bm{V})'_{\theta_m}\bm{V}^{-1}(\bm{V})'_{\theta_n}\right)
    \label{eq:fish_dnf}
\end{align}
where $m,n = 1 \dots, P+2$. In the standard case without measurement error correction, the covariance matrix $\bm{V}$ does not depend on $\bm{\beta}$, meaning that there is no contribution from the trace term for the $\bm{\beta}$ submatrix and off-diagonal ``cross terms" between $\bm{\beta}$ and the traditional variance parameters $\sigma_v^2$ and $\rho$ are 0. In other words, $ \mathcal{I}(\bm{\theta}) = \text{diag}( \mathcal{I}(\bm{\beta}), \mathcal{I}(\sigma_v^2), \mathcal{I}(\rho))$, with slight abuse of notation. For this model, this is not the case, since $\bm{\beta}$ is also present in the covariance matrix. We write 
\begin{align*}
    \mathcal{I}(\bm{\theta}) = \begin{bmatrix}
        \mathcal{I}(\bm{\beta}) & \begin{matrix}
            \mathcal{I}(\bm{\beta}, \sigma_v^2) & \mathcal{I}(\bm{\beta, \rho})
        \end{matrix} \\
        \begin{matrix}
            \mathcal{I}(\sigma_v^2, \bm{\beta}) \\ \mathcal{I}(\rho, \bm{\beta})
        \end{matrix} & \mathcal{I}(\sigma_v^2, \rho)
    \end{bmatrix}
\end{align*}
where $\mathcal{I}(\sigma_v^2, \rho)$, the submatrix involving $\sigma_v^2$ and $\rho$, becomes  
\begin{align*}
\mathcal{I}(\sigma_v^2, \rho) = \begin{bmatrix}
    \rfrac{1}{2}\text{tr}(\bm{V}^{-1}\bm{V}'_{\sigma}\bm{V}^{-1}\bm{V}'_{\sigma}) & \rfrac{1}{2}\text{tr}(\bm{V}^{-1}\bm{V}'_{\sigma}\bm{V}^{-1}\bm{V}'_{\rho}) \\
    \rfrac{1}{2}\text{tr}(\bm{V}^{-1}\bm{V}'_{\rho}\bm{V}^{-1}\bm{V}'_{\sigma}) & \rfrac{1}{2}\text{tr}(\bm{V}^{-1}\bm{V}'_{\rho}\bm{V}^{-1}\bm{V}'_{\rho})
\end{bmatrix}
\end{align*}
by equation~\ref{eq:fish_dnf}, recalling $\bm{V}'_{\sigma}$ and $\bm{V}'_{\rho}$ from Section~\ref{sec:loglik_deriv}. 

Under the assumptions present in this work, the $\bm{\beta}$-submatrix of the Fisher information matrix can be obtained by substituting $(\bm{V})'_{\beta} = 2\text{diag}_{1 \leq d\leq D}(\bm{C}_d\bm{\beta})$ where appropriate, and is
\begin{align*}
    \mathcal{I}(\bm{\beta}) = \hat{\bm{X}}^T\bm{V}^{-1}\bm{e} + 2\sum_{d=1}^D(\bm{V}^{-1})_{dd}(\bm{C}_d\bm{\beta})(\bm{C}_d\bm{\beta})^T.
\end{align*}
For the cross terms, for $p = P+1, P+2$,
\begin{align*}
    \mathcal{I}(\bm{\beta}, \theta_p) &= \frac{1}{2}\text{tr}\left(\bm{V}^{-1}(\bm{V})'_{\beta}\bm{V}^{-1}(\bm{V})'_{\theta_p}\right) = \sum_{d=1}^D\left(\bm{V}^{-1}(\bm{V})'_{\theta_p}\bm{V}^{-1}\right)_{dd}(\bm{C}_d\bm{\beta})
\end{align*}
and
\begin{align*}
    \mathcal{I}(\theta_p, \bm{\beta}) = \frac{1}{2}\text{tr}\left(\bm{V}^{-1}(\bm{V})'_{\theta_p}\bm{V}^{-1}((\bm{V})'_{\beta})^T\right) = \sum_{d=1}^D\left(\bm{V}^{-1}(\bm{V})'_{\theta_p}\bm{V}^{-1}\right)_{dd}(\bm{C}_d\bm{\beta})^T.
\end{align*}
We now describe the Fisher scoring algorithm. Let $\bm{\theta}_0$ denote an initial guess of $\bm{\theta}$. This guess is updated recursively via
\begin{equation}
    \bm{\theta}_{k+1} = \bm{\theta}_{k} + [\mathcal{I}(\bm{\theta}_k)]^{-1}\nabla\ell(\bm{\theta}_k).
    \label{eq:newton}
\end{equation} 
Let $\epsilon$ be some small value and $||\cdot||_2$ the $L_2$ or Euclidean norm. When a sufficiently large $K$ has been reached such that $||\bm{\theta}_{K} - \bm{\theta}_{K-1}||_2 < \epsilon$, we assume that the Newton gradient ascent algorithm has converged and hence $\bm{\theta} \approx \bm{\theta}_K$.  

We now briefly discuss a potential initial guess $\bm{\theta}_0 = (\bm{\beta}_0, (\sigma_v^2)_0, \rho_0)$. The initial guesses for the scalar parameters are application-dependent. Based on some arbitrary choice of $(\sigma_v^2)_0$ and $\rho_0$, we then evaluate 
\[\bm{G}_0((\sigma_v^2)_0, \rho_0) = (\sigma_v^2)_0((\bm{I} - \rho_0\bm{W})(\bm{I} - \rho_0\bm{W}^T))^{-1}, \qquad \tilde{\bm{V}}_0 = \bm{R} + \bm{Z}\bm{G}_0\bm{Z}^T, \]
corresponding to the standard spatial Fay-Herriot case without measurement error \citep{Petrucci06, Rao15}. For $\bm{\beta}_0$ we choose the GLS estimator
\[\bm{\beta}_0 = (\hat{\bm{X}}^T\tilde{\bm{V}}_0^{-1}\hat{\bm{X}})^{-1}\hat{\bm{X}}^T\tilde{\bm{V}}_0^{-1}\Hat{\bm{y}}.\]

\subsection{Mean Squared Error via Parametric Bootstrap}\label{sec:bootstrap}
The mean squared error (MSE) of the EBLUP (see equation \ref{eq:spatial_EBLUPYL}) is given by
\begin{equation}
    \text{MSE}(\Tilde{\bm{y}}) = \mathbb{E}[(\Tilde{\bm{y}} - \hat{\bm{y}})^2].
    \label{eq:mse_definition}
\end{equation}
Approaches such as those presented by \citet{Kackar84, Jiang01, Prasad90} and summarised in Chapter 5.6 of \citet{Rao15}, expand this expectation and derive explicit terms for the MSE. However, these approaches rely on the fact that $\hat{\bm{\beta}}$ is available in closed form as the best linear unbiased estimate (BLUE). While this is the case in the standard spatial EBLUP model \citep{Petrucci06}, we do not use the GLS estimator, as shown in Section~\ref{sec:par_est}. Therefore, we will use a parametric bootstrap approach to directly estimate the expectation in equation~\ref{eq:mse_definition}. This approach was introduced in the Fay-Herriot literature by \citet{Gonzlez-Manteiga08} and is explicitly implemented in the literature in \citet{Moretti20} for a multivariate EBLUP approach.

\begin{algorithm}
\label{algo:par_bootstrap}
Assume that all model assumptions are as before. The following parametric bootstrap procedure allows an estimate for the MSE to be obtained:
    \begin{itemize}
        \item Estimate $\hat{\bm{\theta}} = (\hat{\bm{\beta}}, \hat{\sigma}_v^2, \hat{\rho})$ via equation~\ref{eq:newton}. This requires the covariate data $\hat{\bm{X}}$, the direct estimates $\hat{\bm{y}}$, an initial guess $\bm{\theta}_0$ for the parameter vector and the covariance matrices $\bm{V}$ and $\bm{G}$, which in turn require $\sigma_s^2$, $\bm{W}$ and $\bm{Z}$, all of which are assumed to be known. 
        \item For $b = 1, \dots, B$, do the following:
        \begin{enumerate}
            \item  Write $\hat{\bm{G}} = \hat{\bm{G}}(\hat{\sigma}_v^2, \hat{\rho})$ for the plug-in estimator of $\bm{G}$ using the estimated parameters. Generate bootstrap random effects $\bm{v}_b^* \sim \mathcal{N}(0, \hat{\bm{G}})$, where \newline $\bm{v}_b^* = (v_{1b}^*, \dots, v_{Db}^*)^T$, independently for the small areas $1, \dots, D$. 
            \item Calculate the bootstrap ``true" means $\bm{y}_b^* = \hat{\bm{X}}\hat{\bm{\beta}} + \bm{Z}\bm{v}_b^*$.
            \item Generate bootstrap sampling error $\bm{s}_b^* \sim \mathcal{N}(0, \bm{R})$, where \newline $\bm{s}_b^* = (s_{1b}^*, \dots, s_{Db}^*)^T$, independently for the small areas $1, \dots, D$. Evaluate the bootstrap direct estimate $\Hat{\bm{y}}_b^* = \hat{\bm{X}}\hat{\bm{\beta}} + \bm{Z}\bm{v}_b^* + \bm{s}_b^*$.
            \item Calculate $\tilde{\bm{y}}_b^*$ by substituting $\Hat{\bm{y}}_b^*$ for $\Hat{\bm{y}}$ in equation~\ref{eq:spatial_EBLUPYL}. 
        \end{enumerate}
    \end{itemize}
\end{algorithm}
Using Algorithm~\ref{algo:par_bootstrap}, we estimate the expectation in equation~\ref{eq:mse_definition} by a weighted sum, and the estimated area-specific MSE becomes
\begin{align*}
    \hat{\text{MSE}}(\tilde{\bm{y}}^*) = \frac{1}{B}\sum_{b=1}^B (\tilde{\bm{y}}_b^* - \hat{\bm{y}}_b^*) \odot (\tilde{\bm{y}}_b^* - \hat{\bm{y}}_b^*)
\end{align*}
where $\odot$ is the Hadamard (elementwise) product.

\section{Simulation Study}\label{sec:simulation}

A simulation study is designed to evaluated the performances of our proposed model in comparison to competing models in the literature. 

We compare the following EBLUPs under four competing models:

\begin{enumerate}
    \item Traditional Fay-Herriot model (FH) \citep{FayHerriot79} and its EBLUP. This is denoted by EBLUP in this section.
    \item Spatial Fay-Herriot with SAR(1) process (SFH) \citep{pratesi2008small} and with EBLUP denoted by SEBLUP.
    \item Ybarra-Lohr Fay-Herriot model (FH-YL) \citep{Ybarra08} and with EBLUP denoted by EBLUP-YL.
    \item Our proposal, the spatial Ybarra-Lohr Fay-Herriot with EBLUP denoted by SEBLUP-ME.
\end{enumerate}

Let $D=100$ denote the number of small areas. The steps of the simulation are as follows:

\begin{enumerate}
    \item Generate  $x_{d1} \sim \mathcal{N}(5, \sqrt{\sigma^{2}_{d1}})$, and $x_{d2} \sim \mathcal{N}(4.5,  \sqrt{\sigma^{2}_{d2}})$   with $\sigma^{2}_{d1} \sim \mathcal{N}(0.5, 0.1)$ and $\sigma^{2}_{d2} \sim \mathcal{N}(0.3, 0.05)$, hence the vector of observed area-level covariates (including the intercept) is $\mathbf{x}_{d}=[1\,\,\, x_{d1}\,\,\,x_{d2}]$. The regression coefficients are $\bm{\beta} =
[1.5\,\,\,1.5\,\,\,0.3]$. Set $\bm{Z} = \bm{I}$.

    \item Repeat $R=500$ times $(r=1,...,500)$:

    \begin{enumerate}[label=\theenumi.\arabic*]
        \item Generate $v^{(r)} \sim \mathcal{N}(0,\bm{G})$  where $ v^{(r)} $ follows a SAR(1) process with covariance matrix $\bm{G} = \sigma_v^2((\bm{I} - \rho\bm{W})(\bm{I} - \rho\bm{W}^T))^{-1}$, with $\bm{W}$ a known weight matrix as before. We set $\sigma^{2}_v = 0.2$ and $\rho = 0.25$. For the design error, $s^{(r)}_{d} \sim \text{Unif}(0.02, 0.22) $ and hence
        $\hat{y}^{(r)}_{d} = \bm{x}_{d}\bm{\beta} + v^{(r)}_{d} + s^{(r)}_{d}$. The true parameter (population mean) is $y^{(r)}_{d}=\bm{x}_{d}\bm{\beta} + v^{(r)}_{d}$.
        In doing so, we assess three measurement error scenarios:

        \begin{itemize}
            \item $x_{d1}$ is correctly observed (generated in step 1 and kept fixed over the iterations)
            \item $x_{d1}$ is affected by classical measurement error: $x^{(r)}_{d1} = x_{d1} + \mathcal{N}(0, 0.1)$
            \item $x_{d1}$ is affected by systematic error: $x^{(r)}_{d1} = x_{d1} + \mathcal{N}(3, 0.1)$.
        \end{itemize}

        \item For every measurement error scenario, compute the EBLUP of $y^{(r)}_{d}$ under the different models for $d=1,...,D$.

    \end{enumerate}

\item  Compute the percentage relative root mean squared error (RRMSE) and the absolute relative bias (ARB) as performance measures, where $\hat{y}_d$ denotes an EBLUP, under a specific area-level model, for $y_d$:  

\begin{equation}
    \text{RRMSE} (\hat{y}_d)
=
\frac{
\sqrt{
\frac{1}{R}
\sum_{r=1}^R
(y^{(r)}_d-\hat{y}^{(r)}_d)^2
}
}{
\frac{1}{R}
\sum_{r=1}^R y^{(r)}_d
} \times 100,
\end{equation}
    
\end{enumerate}

\begin{equation}
\text{ARB} (\hat{y}_d)
=
\left|
\frac{
\frac{1}{R}
\sum_{r=1}^{R}
\left(\hat{y}^{(r)}_{d}-y^{(r)}_d\right)
}{
\frac{1}{R}  \sum_{r=1}^{R} y^{(r)}_d
}
\right| \times 100.
\end{equation}

These are common measures in model-based simulation studies in SAE (see, e.g., \cite{berg2023empirical}). In the simulation study section results, we also present descriptive statistics of those across the small areas.

\subsection{Simulation Results}

\subsubsection*{No measurement error in the covariates}

We first consider the baseline scenario in which the auxiliary variables are observed without measurement error. This setting serves as a benchmark, allowing us to check that the proposed model does not lead to a loss in efficiency when the additional measurement error correction is unnecessary. We present in Figures \ref{RRMSE_noME} and \ref{ARB_noME} the boxplots of the RRMSE and ARB, respectively, of the small area estimates computed under the different models, when the $(x_{11}, \dots, x_{D1})$ are correctly observed (no measurement error). In Figures \ref{RRMSE_noME_ord} and \ref{ARB_noME_ord}, we show the RRMSE and ARB respectively across the small areas, ordered by growing sampling variance of the direct estimates.  

As expected, all four estimators exhibit very similar performance. The distributions of the RRMSE almost completely overlap, indicating that accounting for measurement error does not reduce efficiency when the auxiliary variables are correctly observed. Likewise, all EBLUB approaches display negligible bias, with median ARB values close to zero and only limited variability across the simulation replicates. The area-specific results further confirm that accuracy remains comparable irrespective of the modelling approach, even for areas characterised by relatively large sampling variances.

These findings are reassuring, as they demonstrate that the proposed SEBLUP-ME approach does not introduce unnecessary variability when the assumptions of the classical spatial Fay-Herriot model hold. In other words, the additional covariance structure incorporated in the proposed model does not  affect the properties of the small area estimates in the absence of measurement error.

\begin{figure}[H]
\centering
\includegraphics[width=0.55\textwidth]{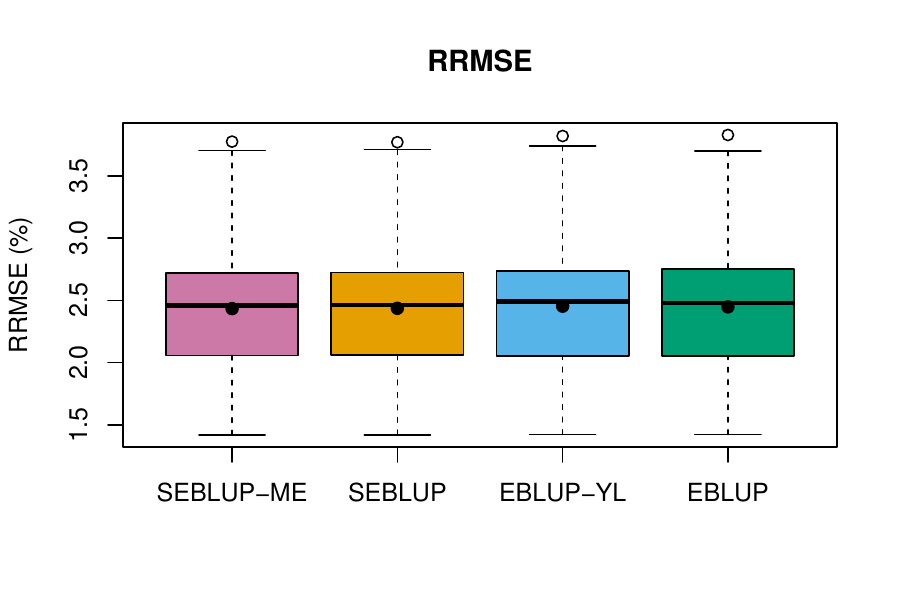}
\caption{RRMSE comparison of the four estimators - scenario with no measurement error. The dot denotes the mean across the areas.}
\label{RRMSE_noME}
\end{figure}

\begin{figure}[H]
\centering
\includegraphics[width=0.55\textwidth]{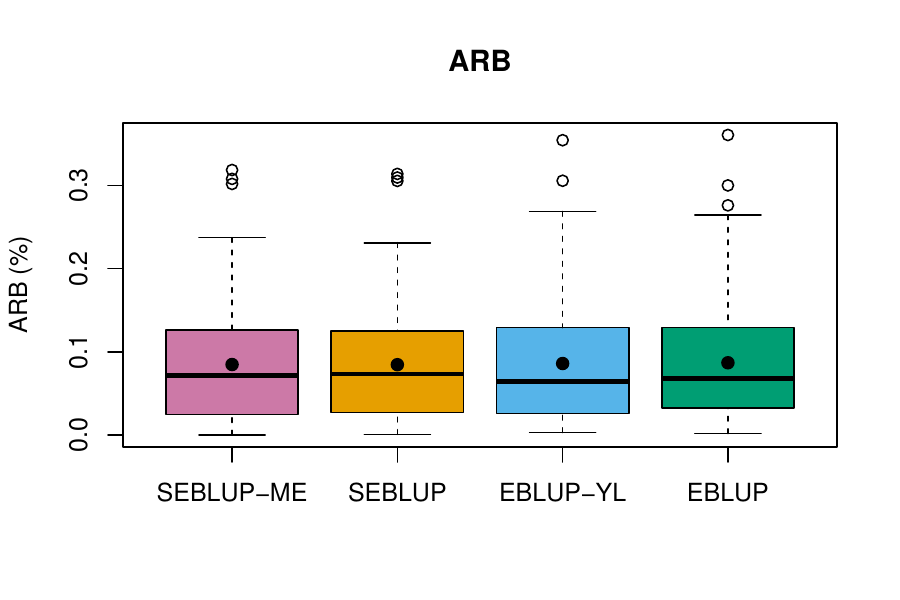}
\caption{ARB comparison of the four estimators - scenario with no measurement error. The dot denotes the mean across the areas.}
\label{ARB_noME}
\end{figure}

\begin{figure}[H]
\centering
\includegraphics[width=0.75\textwidth]{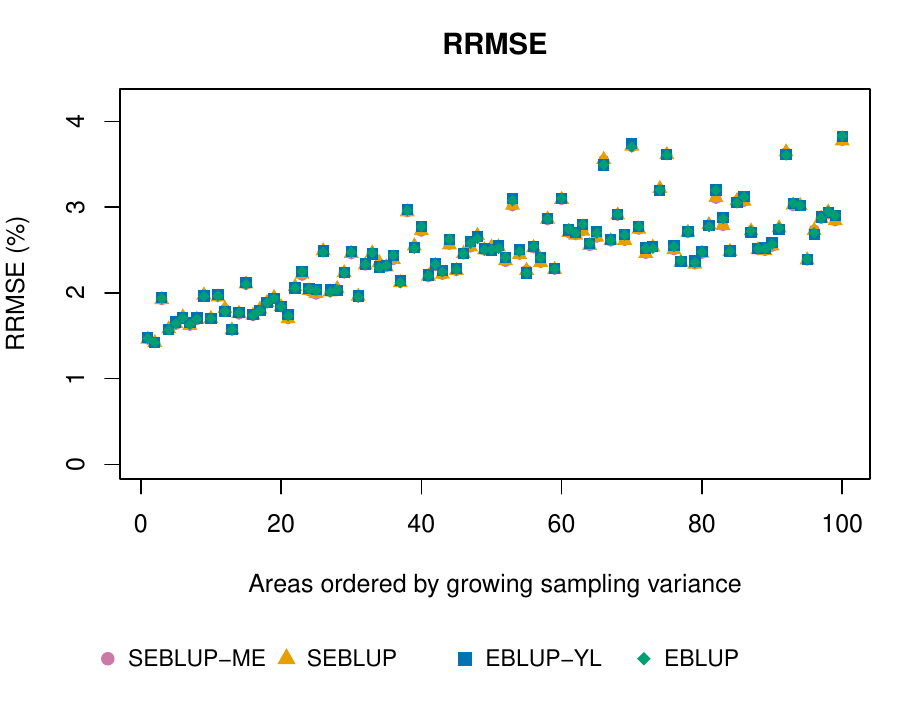}
\caption{RRMSE comparison across the areas (ordered by growing $\sigma^2_s$) of the four estimators - scenario with no measurement error. }
\label{RRMSE_noME_ord}
\end{figure}

\begin{figure}[H]
\centering
\includegraphics[width=0.75\textwidth]{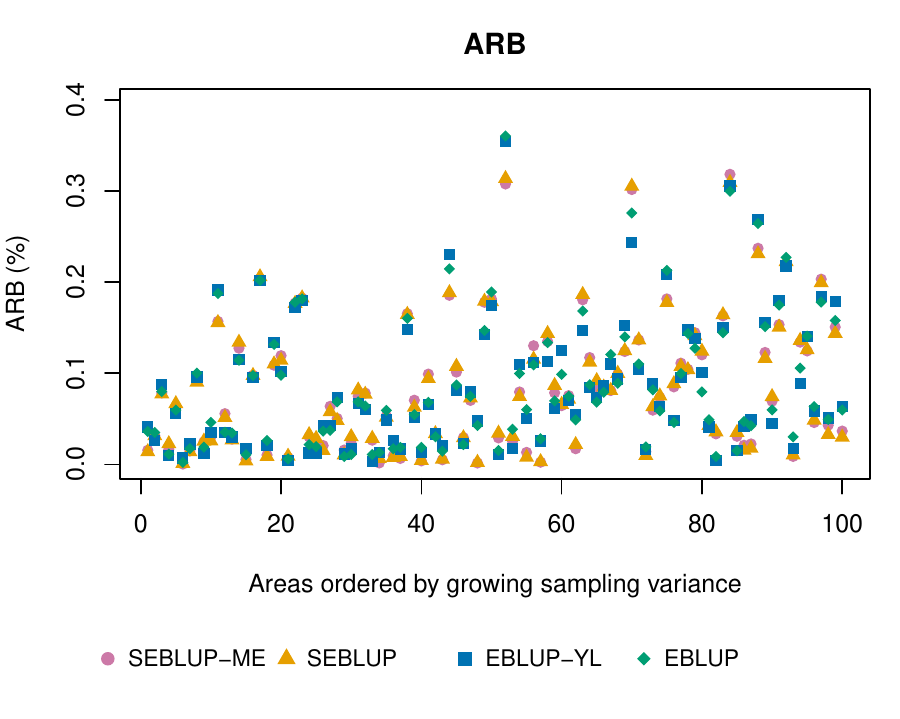}
\caption{ARB comparison across the areas (ordered by growing $\sigma^2_{s}$) of the four estimators - scenario with no measurement error. }
\label{ARB_noME_ord}
\end{figure}

\subsubsection*{Classical measurement error in one covariate}

The second scenario introduces classical measurement error in the first auxiliary variable, representing random inaccuracies that may arise when covariates are obtained from non-traditional data sources. These inaccuracies are assumed to have zero mean. In this setting, important differences emerge between the competing estimators. In Figures \ref{RRMSE_CME} and \ref{ARB_CME}, we depict the boxplots of the RRMSE and ARB, respectively, of the small area estimates computed under the different models when the first covariate is affected by classical measurement error.  In Figures \ref{RRMSE_CME_ord} and \ref{ARB_CME_ord}, we present the RRMSE and ARB across areas, ordered by growing sampling variance, as in the previous subsection. 

Ignoring measurement error leads to a noticeable worsening in  accuracy. Both the traditional Fay–Herriot model (EBLUP) and the spatial Fay–Herriot model (SEBLUP) exhibit larger RRMSE values than the models explicitly accounting for uncertainty in the auxiliary variable. Although the spatial model benefits from borrowing strength across neighbouring areas, the gain obtained through the spatial process alone is insufficient to compensate for the bias introduced by the covariate affected by measurement error. The Ybarra-Lohr model, which explicitly models classical measurement error but ignores spatial dependence, substantially improves upon both the classical and spatial Fay–Herriot estimators.

\begin{figure}[H]
\centering
\includegraphics[width=0.55\textwidth]{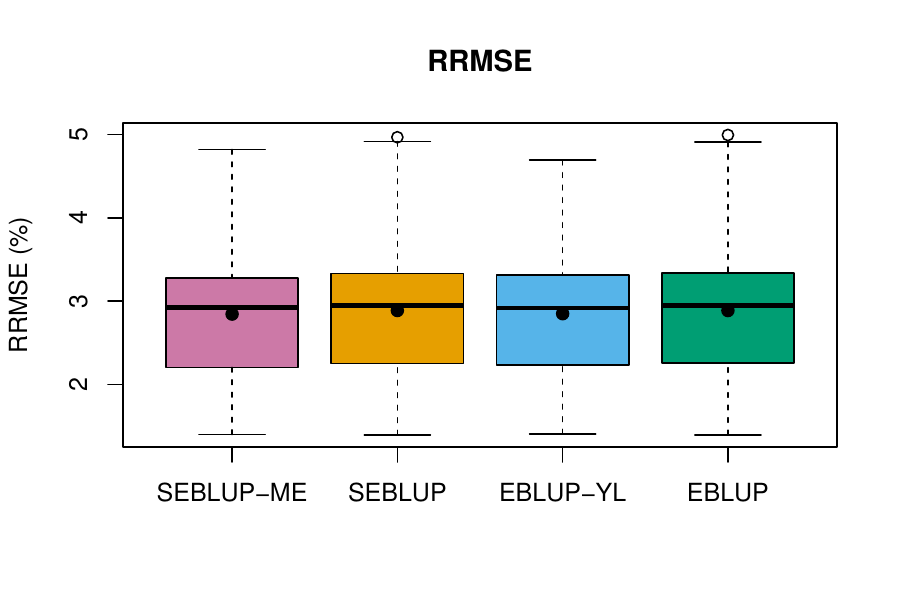}
\caption{RRMSE comparison of the four estimators - scenario with classical measurement error in $X_1$. The dot denotes the mean across the areas.}
\label{RRMSE_CME}
\end{figure}

\begin{figure}[H]
\centering
\includegraphics[width=0.55\textwidth]{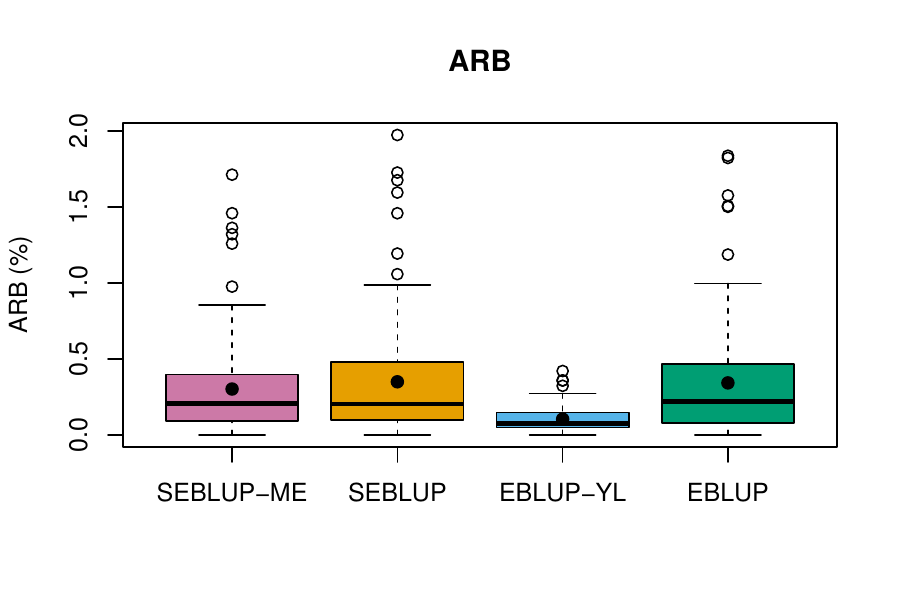}
\caption{ARB comparison of the four estimators - scenario with classical measurement error in $X_1$. The dot denotes the mean across the areas.}
\label{ARB_CME}
\end{figure}

\begin{figure}[H]
\centering
\includegraphics[width=0.75\textwidth]{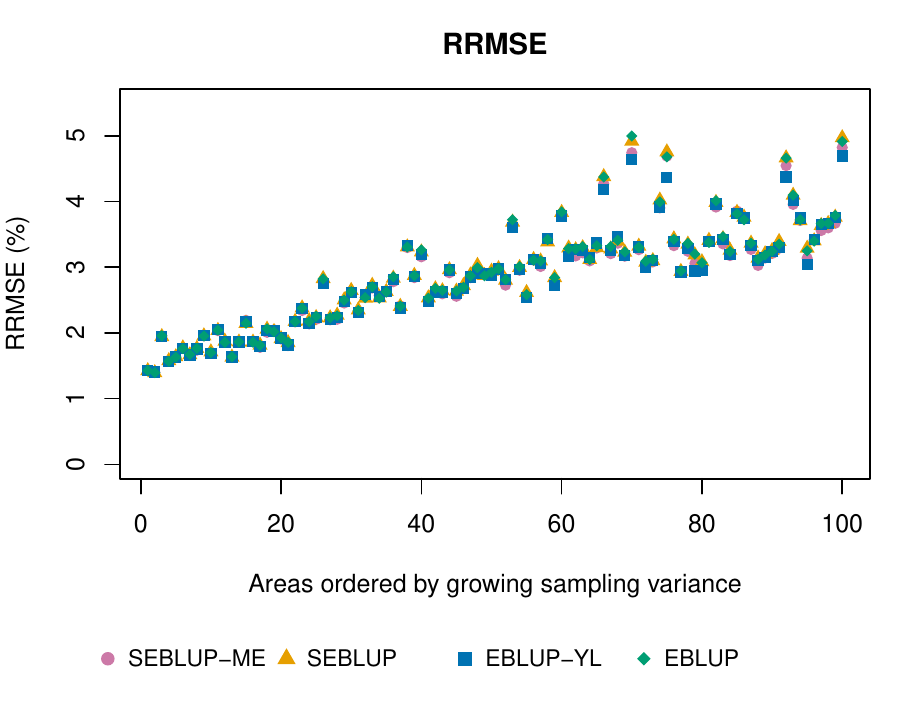}
\caption{RRMSE comparison across the areas (ordered by growing $\sigma^2_s$) of the four estimators - scenario with classical measurement error in $X_1$. }
\label{RRMSE_CME_ord}
\end{figure}

\begin{figure}[H]
\centering
\includegraphics[width=0.75\textwidth]{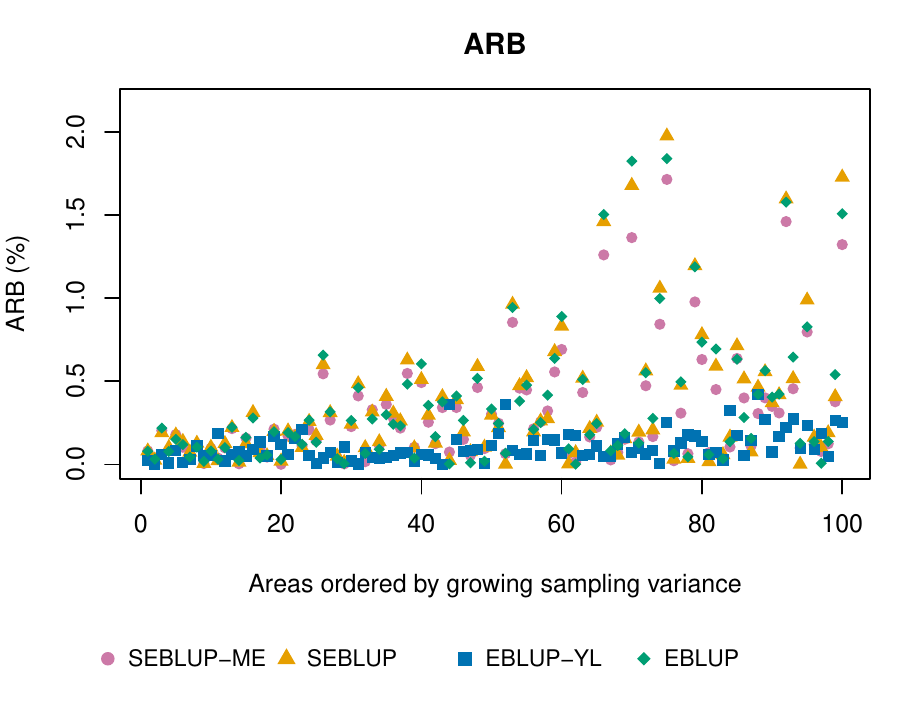}
\caption{ARB comparison across the areas (ordered by growing $\sigma^2_s$) of the four estimators - scenario with classical measurement error in $X_1$. }
\label{ARB_CME_ord}
\end{figure}

\subsubsection*{Systematic measurement error in the first covariate}

The third scenario represents a considerably more challenging situation by introducing systematic measurement error in the first covariate. Unlike classical measurement error, systematic error introduces a persistent (systematic) shift in the observed covariate, closely resembling the types of bias encountered when auxiliary variables are derived from non-probability data sources. In Figures \ref{RRMSE_SME} and \ref{ARB_SME}, we depict the boxplots of the RRMSE and ARB, respectively, of the small area estimates computed under the different models, when the first covariateis affected by systematic measurement error.  In Figures \ref{RRMSE_SME_ord} and \ref{ARB_SME_ord}, we present the RRMSE and ARB across areas, ordered by growing sampling variance, as previously done. 

Under this scenario, the shortcomings of models that ignore measurement errors become substantially more pronounced. The spatial Fay–Herriot model provides only partial improvement through spatial borrowing.  Here, the Ybarra–Lohr model produces estimates with larger ARB.  Across all performance measures, the proposed SEBLUP-ME estimator demonstrates the best behaviour. It produces the lowest ARB whilst returning small area estimates with an RRMSE smaller than common thresholds of Official Statistics for estimates releasing (see, e.g.,  \cite{spagnolo2018}). These results suggest that modelling spatial dependence and systematic measurement error simultaneously provides important practical benefits of our approach.

\begin{figure}[H]
\centering
\includegraphics[width=0.55\textwidth]{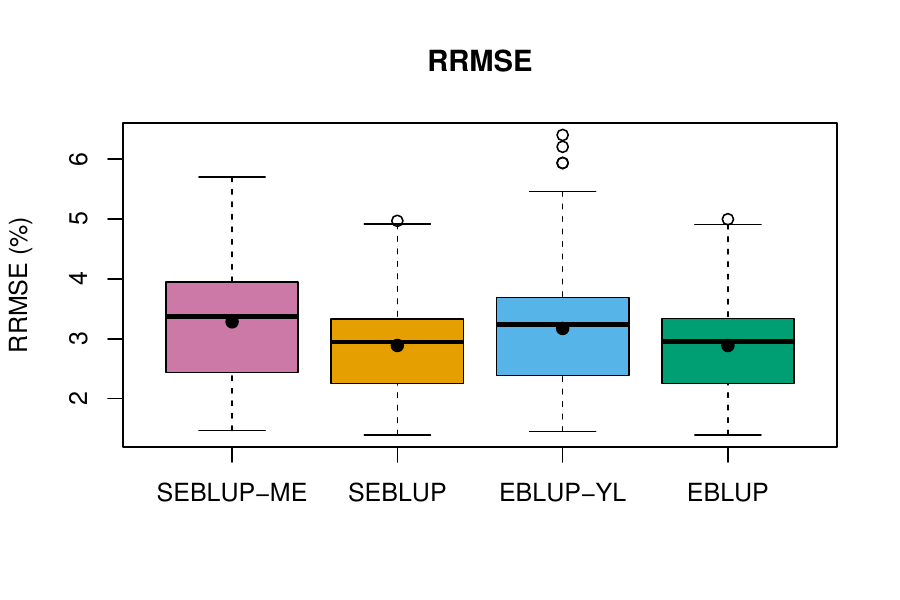}
\caption{RRMSE comparison of the four estimators - scenario with classical measurement error in $X_1$. The dot denotes the mean across the areas.}
\label{RRMSE_SME}
\end{figure}

\begin{figure}[H]
\centering
\includegraphics[width=0.55\textwidth]{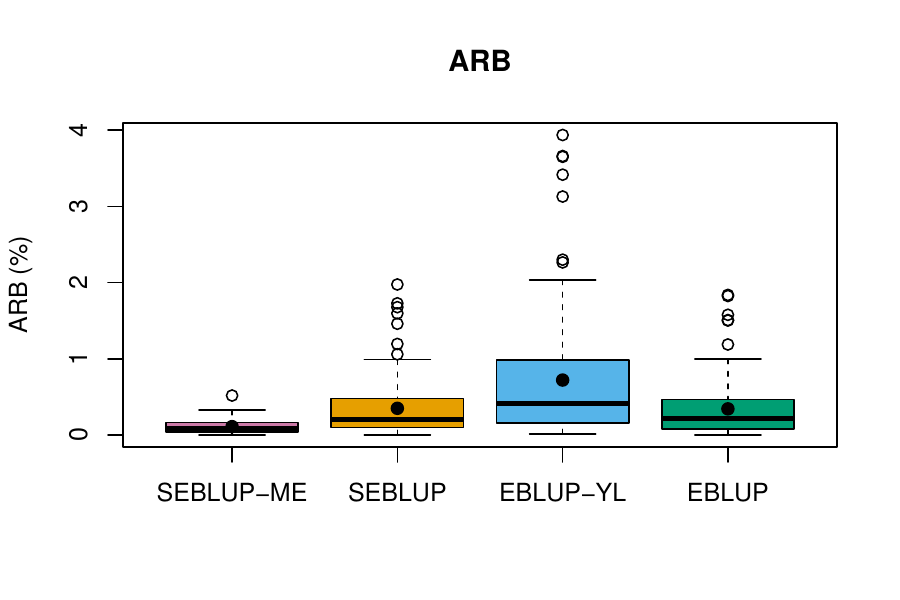}
\caption{ARB comparison of the four estimators - scenario with systematic measurement error in $X_1$. The dot denotes the mean across the areas.}
\label{ARB_SME}
\end{figure}

\begin{figure}[H]
\centering
\includegraphics[width=0.75\textwidth]{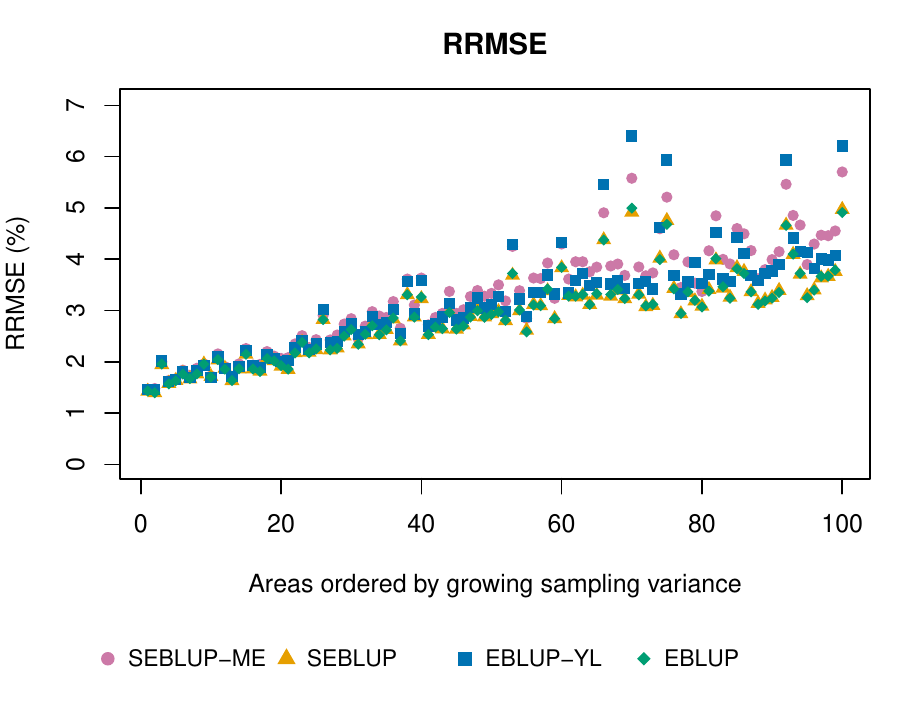}
\caption{RRMSE comparison across the areas (ordered by growing $\sigma^2_s$) of the four estimators - scenario with systematic measurement error in $X_1$. }
\label{RRMSE_SME_ord}
\end{figure}

\begin{figure}[H]
\centering
\includegraphics[width=0.75\textwidth]{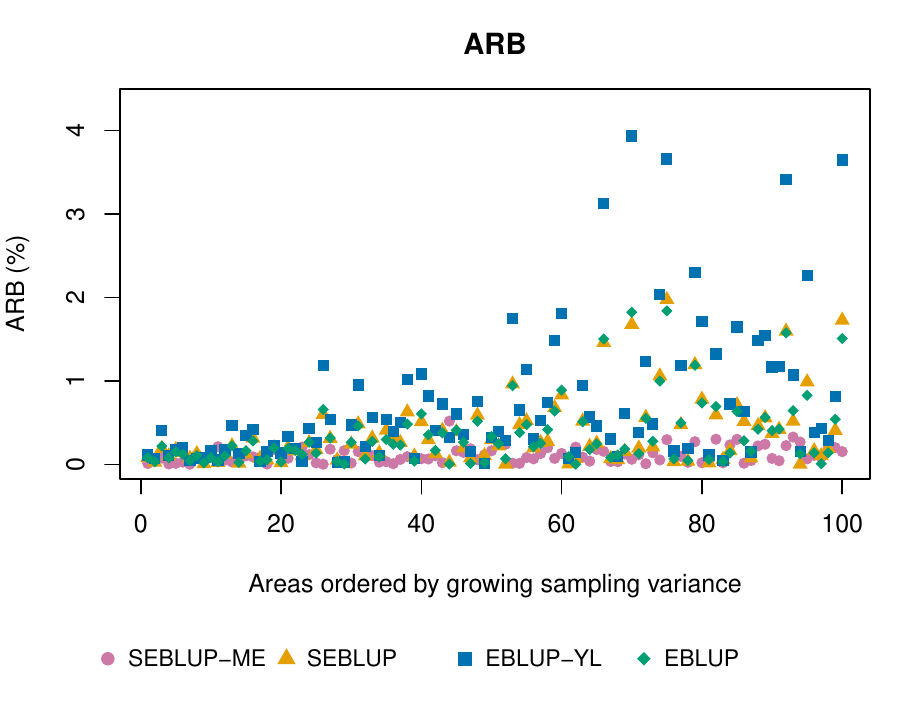}
\caption{ARB comparison across the areas (ordered by growing $\sigma^2_s$) of the four estimators - scenario with systematic measurement error in $X_1$. }
\label{ARB_SME_ord}
\end{figure}

\subsection{Final remarks on the simulation study}

The scenario involving systematic measurement error represents a deliberate departure from the assumptions underlying both the Ybarra--Lohr model and the proposed EBLUP (EBLUP-ME). Neither estimator is expected to fully eliminate the bias introduced by a systematic error in the covariates. Nevertheless, the proposed SEBLUP-ME consistently exhibits lower ARB than the Ybarra--Lohr approach. This improvement should not be interpreted as a correction for systematic measurement error per se, but rather as a consequence of incorporating spatial dependence into the prediction model. By borrowing strength from neighbouring areas through the introduction of a SAR(1) process in the area-level model, EBLUP-ME relies less heavily on the biased regression component on its own and exploits the spatial correlation present in the data generating process. This additional source of information partially compensates for the misspecification induced by the biased auxiliary variable, resulting in more stable predictions and a smaller average prediction bias across the areas. The gains are particularly evident for areas with larger sampling variances, where the contribution of the spatial random effects to the prediction is greatest.

\section{Application}\label{sec:application}  
\subsection{Data, variables, and small area problem}
In this section, we propose an application based on the case study in \cite{salvatore2024use}. 
Using data from Round 10 of the European Social Survey (ESS), we aim to examine worry about climate change in Spain at the NUTS2 regional level. The analysis includes 17 autonomous regions (excluding Ceuta and Melilla due to missing auxiliary data), with a total sample of 2,214 respondents. The ESS is not designed to produce reliable subnational estimates in Spain. The outcome of interest is the proportion of individuals who are very or extremely worried about climate change. Respondents rated their worries on a 5-point scale, which, following \cite{salvatore2024use}, was dichotomised into ``very/extremely worried'' (scores 4–5) versus ``other'' (scores 1–3) using a median split approach (median equal to 4). At the national level, 56.4\% of respondents were classified as highly worried about climate change. Direct regional estimates showed considerable variability, and many regions had coefficients of variation (CVs) exceeding 16\%, indicating that the estimates are unreliable and further justifying the use of SAE techniques \citep{marchetti2024social}.

We use web-scraped data from the platform Booking.com to capture regional sustainability characteristics. Specifically, we created variables showing the proportion of hospitality venues classified under the platform’s former Travel Sustainable program, which included four sustainability levels (Level 1, Level 2, Level 3, and Level 3+). Although the program was replaced in March 2024 by a single Sustainability Certification label, the data were collected before this change and therefore use the original classification system. The underlying assumption of the choice of these variables is that a higher proportion of environmentally certified venues reflects greater regional environmental awareness and sustainability practices, making these indicators potentially useful predictors of regional attitudes towards climate change in Spain \citep{salvatore2024use}. Data were collected through web scraping in R using the ``rvest'' package. For each Spanish region, information was gathered on the total number of venues and the proportion assigned to each sustainability level. To account for fluctuations in availability, data were collected for 84 randomly selected dates (approximately seven per month) between February 2024 and January 2025, and regional proportions were averaged across this period. Across all regions, the average proportions of hotels were 17.2\% for Level 1, 9.26\% for Level 2, 2.74\% for Level 3, and 2.11\% for Level 3+.

Given the promising findings reported in \cite{salvatore2024use}, we employ these non-traditional variables to evaluate and compare the performance of the models considered in the simulation study. The Booking.com variables are derived from web-scraped data and are therefore subject to measurement error due to temporal variation in listings and the data collection process itself. Furthermore, climate change attitudes are likely to display spatial dependence across neighbouring regions. These arguments motivates the use of a spatial Fay--Herriot model with measurement error. 

\subsection{Mapping worry about climate change at NUTS-2 level in Spain}

The small area estimates are compared via the RRMSE in Figure \ref{fig:rrmse_app} (we refer to Sections \ref{sec:FH} and \ref{Sec:spatialFH} for the MSE estimators), and coefficient of variation (CV) for the direct estimates, as is standard practice in SAE \citep{Rao15}. Furthermore, as a bias diagnostic, we present the ratios between the EBLUPs and the direct estimates in Figure \ref{fig:ratios_app}, given that the direct estimates are design-unbiased. 
We find that the estimated MSE for the new measurement error SEBLUP model is higher than the estimated MSE of the other models, though it is within the 16.5\% threshold. A possible reason for this is that the other models underestimate the MSE due to a lack of awareness of errors in the covariates. In Figure~\ref{fig:ratios_app}, we see that the prediction aligns more closely with the direct estimates than the other models do. The SEBLUP-ME model seems to assign more weight to the direct estimates when covariate data are unreliable.

 We present in Figure \ref{fig:map}, the regional estimates for Spain at NUTS2 of the worry about climate change indicator via the SEBLUP-ME approach. The estimates range from approximately 0.43 to 0.64, indicating substantial geographic variation in worry about climate change. A clear spatial pattern emerges. The highest estimated levels of worries (dark purple regions) are observed in several peripheral regions, particularly in the north-west, north-east, and parts of the south, where values exceed 0.58. In contrast, the lowest levels of concern (yellow regions) are concentrated in the central part of Spain, with estimates around 0.44–0.47. Most regions fall within an intermediate range (orange to pink shades), suggesting moderate worry about climate change. The Canary Islands also appear to exhibit relatively high levels of concern compared with the national average.

Overall, the results indicate that worry about climate change are not uniformly distributed across Spain, with residents in peripheral regions generally expressing greater worry than those in the central regions. The difference between the highest and lowest regional estimates is approximately 0.21 points, highlighting meaningful regional heterogeneity in climate change concern.

\begin{figure}[H]
    \centering
    \includegraphics[width=.75\textwidth]{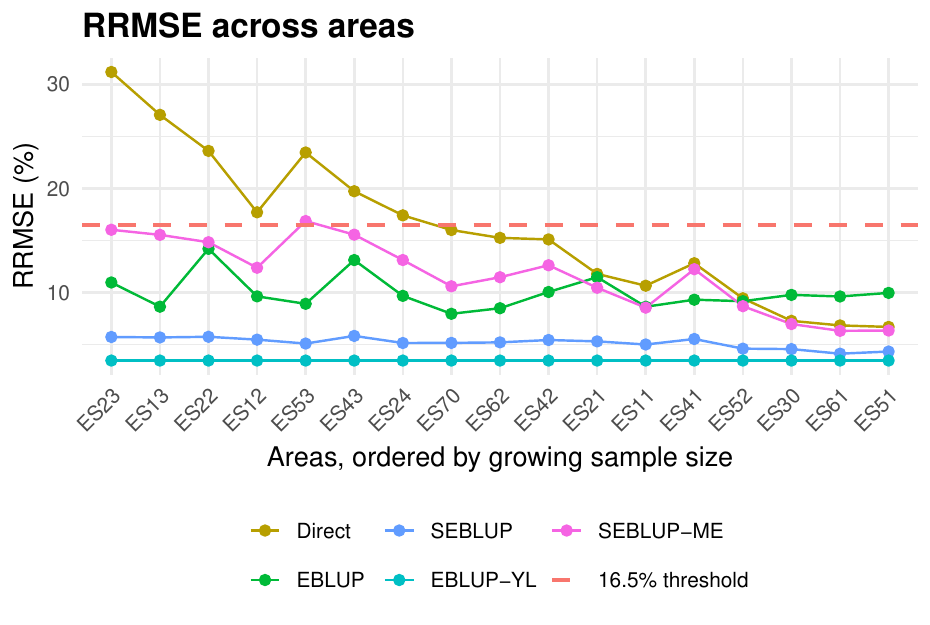}
    \caption{Estimated RRMSE of the different EBLUPs (and CV of the direct estimates), ordered by increasing population. These are estimated via parametric bootstrap.}
    \label{fig:rrmse_app}
\end{figure}

\begin{figure}[H]
    \centering
    \includegraphics[width=.75\textwidth]{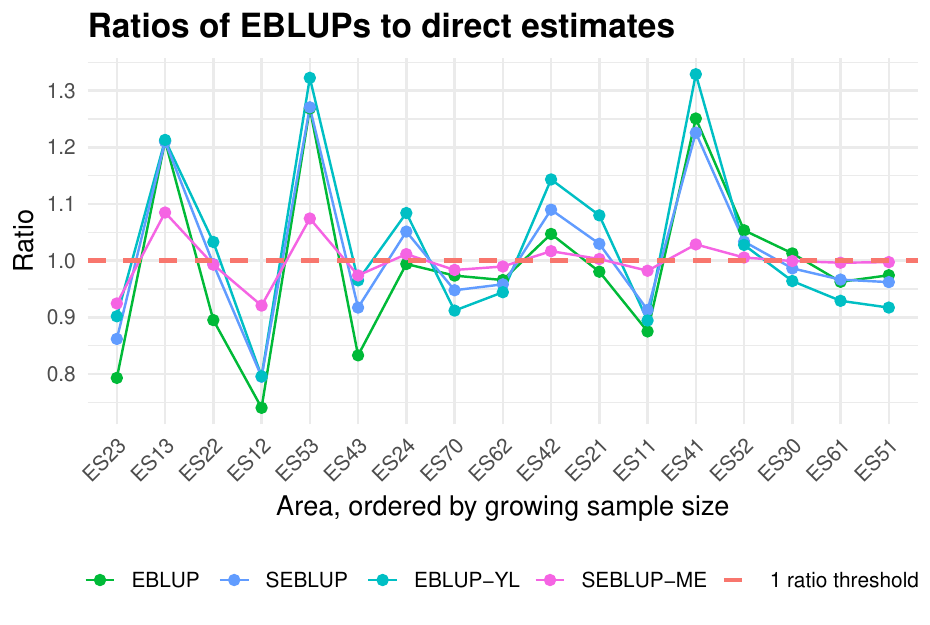}
    \caption{Ratios between the different EBLUPs and the direct estimates, ordered by growing sample size.}
    \label{fig:ratios_app}
\end{figure}

\begin{figure}[H]
    \centering
    \includegraphics[width=.8\textwidth]{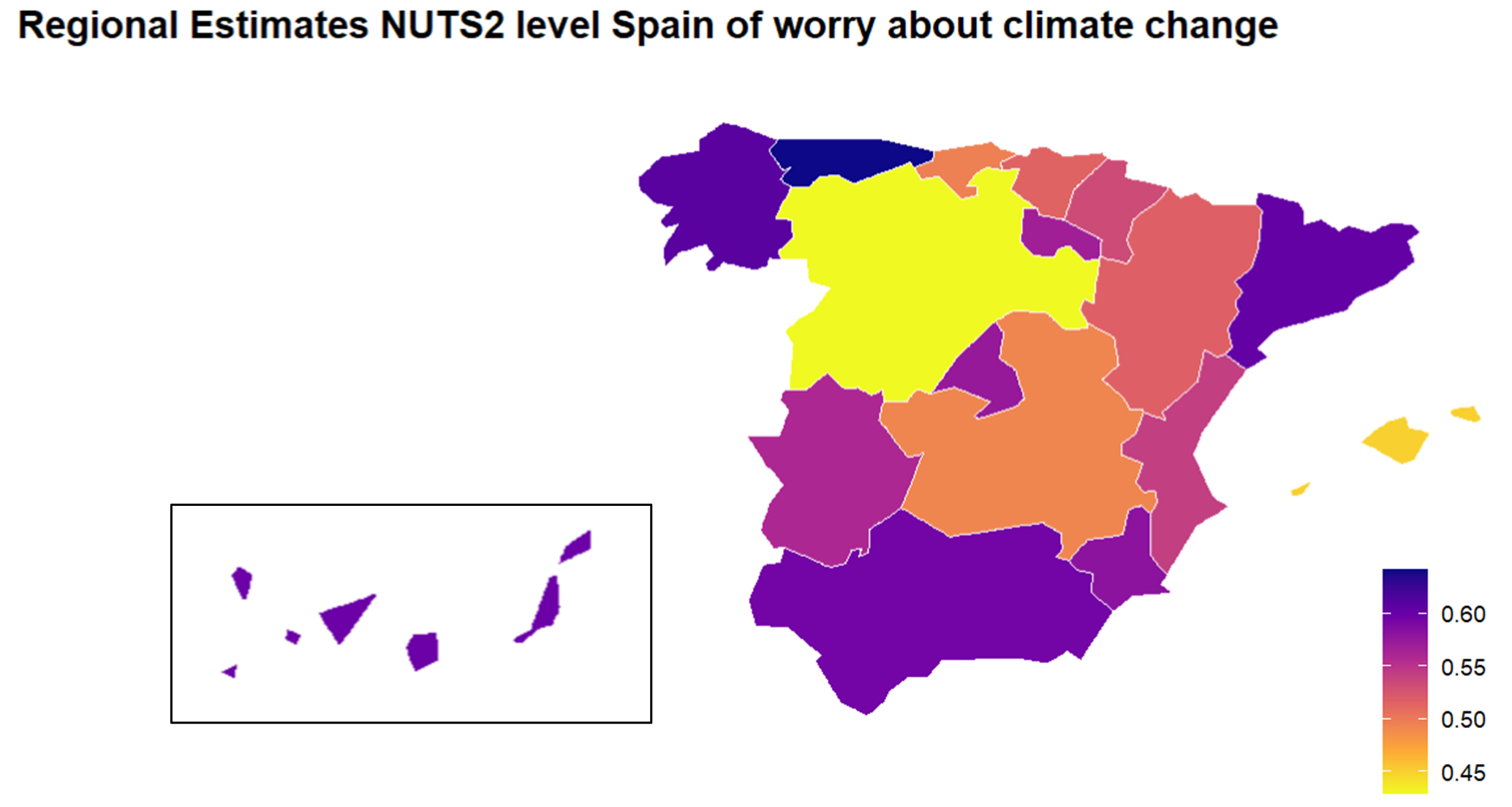}
    \caption{Regional SEBLUP-ME estimates at NUTS2 level for Spain of worry about climate change indicator. Darker colours denote more worry about climate change.}
    \label{fig:map}
\end{figure}

\section{Conclusion}\label{sec:conclusions}

This article introduced a spatial Fay--Herriot model with measurement error in the covariates, extending the framework of \citet{Ybarra08} to accommodate spatial dependence between areas via a SAR(1) process. We derived the EBLUP under this model, established the associated log-likelihood and score equations, proposed a Fisher scoring algorithm for parameter estimation, and outlined a parametric bootstrap procedure for MSE estimation. Together, these results provide practitioners with an approach to handle complex modelling situations that often arise when auxiliary variables are derived from Big Data.

Our simulation study shows that the proposed SEBLUP-ME behaves as intended across a range of measurement error scenarios. When covariates are observed without error, the additional covariance structure introduces no meaningful loss of efficiency relative to the classical spatial EBLUP, confirming that the correction is not costly when unneeded. Under classical (zero-mean) measurement error, the correction proposed by \citet{Ybarra08} is the dominant source of improvement, and our spatial extension performs comparably to the non-spatial Ybarra--Lohr approach. The results are most notable in the systematic (non-zero-mean) measurement error scenario, a scenario that closely mirrors the kind of bias typically introduced by non-probability, Big Data-derived covariates. Here, the proposed spatial Ybarra--Lohr extension tends to outperform the classical Ybarra--Lohr estimator, because systematic distortion in the auxiliary variables reduces the reliability of the synthetic predictions. Under non-classical measurement error, i.e., $\hat{x}_{di} = x_{di} + u_d, \mathbb{E}[u_d] \neq 0$ for covariate $i = 1, \dots, P$ and small area $d = 1, \dots, D$, the observed auxiliary information is not only noisy but also systematically shifted away from the true covariate values, so that the regression component may become substantially biased for individual areas. The classical Ybarra--Lohr estimator corrects for measurement error only through the additional variance component $\bm{\beta}^T \bm{C}_d \bm{\beta}$, which accounts for the extra uncertainty induced by contaminated covariates, but the model still assumes independent area random effects and therefore relies primarily on local auxiliary information for prediction. The estimator proposed here additionally incorporates a SAR(1) spatial process in the area random effects, inducing spatial dependence across neighbouring areas, so that estimation is no longer driven solely by the contaminated covariates of a single area but also by information borrowed from surrounding areas through the spatial covariance structure. When the auxiliary variables are systematically distorted, as is plausible with Big Data sources, neighbouring areas may still contain useful information regarding the underlying latent process; the spatial component therefore stabilises estimation and mitigates the effect of biased auxiliary information, yielding lower bias than the classical Ybarra--Lohr estimator in the presence of non-zero-mean measurement error.

The application to European Social Survey data on climate change worry in Spain, using web-scraped Booking.com sustainability indicators as auxiliary information, illustrates the practical relevance of these findings. The SEBLUP-ME estimates achieved RRMSE values below common Official Statistics thresholds for essentially all regions, including several with small sample sizes and correspondingly unreliable direct estimates, while the ratio diagnostics indicated no systematic departure from the design-unbiased direct estimates. 

Several limitations point to directions for future work. First, we assumed that the measurement error covariance matrices $\bm{C}_d$ are known; in practice these must often be estimated from the Big Data source itself, and propagating this additional layer of uncertainty into the EBLUP and its MSE remains an open problem. Secondly, we restricted attention to a SAR(1) specification for spatial dependence; extensions to more general spatial processes, temporal dynamics, or multivariate outcomes would broaden the applicability of the model. Finally, finding analytical approximations for the MSE of the proposed EBLUP remains interesting from a theoretical perspective, but is beyond the scope of this article. Future research will study this problem. Overall, our model yields good performance and does not worsen the performance of competing methods under any of the scenarios considered, with biases remaining negligible across all areas. In practice, we recommend that users begin with the simplest model that adequately represents the data and extend it only when there is evidence that a more complex specification is needed.

\subsubsection*{Data availability statement}

The data used in this study are publicly available from the European Social Survey (ESS) Round 10. The dataset can be accessed, subject to registration, through the European Social Survey data portal (DOI: 10.21338/ESS10E03\_2; previously archived under DOI: 10.21338/NSD-ESS10-2020).

\subsubsection*{Funding statement}
This work was supported by the Dutch Research Council (NWO) under grant number 406.XS.25.01.042.

\subsubsection*{Conflict of interest statement} The authors declare that they have no conflicts of interest relevant to this work.


\renewcommand*{\bibfont}{\footnotesize}
\bibliographystyle{agsm} 

\bibliography{references}

\end{document}